\documentclass[conference]{IEEEtran}
\IEEEoverridecommandlockouts
\usepackage{grffile,tikz}
\usepackage{amsmath,subfig,xspace,bbold,amsthm,multirow,graphicx}
\usepackage[pointedenum]{paralist}
\usepackage{listings}
\usepackage{enumitem}
\usepackage[ruled]{algorithm2e}
\usepackage[breaklinks=true,pdfstartview=FitH]{hyperref}

\newcommand{\citep}[1]{\cite{#1}}
\newcommand{\arxiv}[1]{{#1}}
\newcommand{\canomit}[1]{{}}
\newtheorem{theorem}{Theorem}
\newtheorem{corollary}{Corollary}
\newtheorem{lemma}{Lemma}
\def\R{{ \mathbf{R}}}
\def\N{{ \mathbf{N}}}

\def\bg{{\boldsymbol g}}

\def\bx{{\boldsymbol x}}
\def\by{{\boldsymbol y}}

\def\bz{{\boldsymbol z}}

\def\bp{{\boldsymbol p}}

\def\bs{{\boldsymbol s}}

\def\bd{{\boldsymbol d}}
\def\b1{{\boldsymbol 1}}

\DeclareMathOperator{\trace}{tr}
\def\sjwcomment#1{\footnote{{\bf SJW:} #1}}
\newcommand{\modify}[1]{{{#1}}}
\newcommand{\smodify}[1]{{{#1}}}

\begin{document}
\title{Using Neural Networks to Detect Line Outages from PMU
	Data\thanks{This work was supported by a DOE grant subcontracted
	through Argonne National Laboratory Award 3F-30222,  National
Science Foundation Grants IIS-1447449 and CCF-1740707, and AFOSR Award
FA9550-13-1-0138.}
}
\author{Ching-pei Lee and Stephen J. Wright\thanks{C. Lee and S. J. Wright are with the Computer Sciences Department, 1210 W. Dayton Street, University of Wisconsin, Madison, WI 53706, USA (e-mails: \url{ching-pei@cs.wisc.edu} and \url{swright@cs.wisc.edu}).}
}
\maketitle
\begin{abstract}
We propose an approach based on neural networks and the AC power flow
equations to identify single- and double-line outages in a power grid
using the information from phasor measurement unit sensors
(PMUs) \modify{ placed on only a subset of the buses.}
Rather than inferring the outage from the sensor data by
inverting the physical model, our approach uses the AC model to
simulate sensor responses to all outages of interest under multiple
demand and seasonal conditions, and uses the resulting data to train a
neural network classifier to recognize and discriminate between
different outage events directly from sensor data. After training,
real-time deployment of the classifier requires just a few
matrix-vector products and simple vector operations. These operations
can be executed much more rapidly than inversion of a model based on
AC power flow, which consists of nonlinear equations and possibly
integer / binary variables representing line outages, as well as the
variables representing voltages and power flows. We are motivated to
use neural network by its successful application to such areas as
computer vision and natural language processing. Neural networks
automatically find nonlinear transformations of the raw data that
highlight useful features that make the classification task easier.
\modify{We describe a principled way to choose sensor locations and
  show that accurate classification of line outages can be achieved
  from a restricted set of measurements, even over a wide range of
  demand profiles.}
\end{abstract}

\section{Introduction}
Phasor measurement units (PMUs) have been introduced in recent years
as instruments for monitoring power grids in real time.  PMUs provide
accurate, synchronized, real-time information of the voltage phasor at
30-60 Hz, as well as information about current flows. When processed
appropriately, this data has the potential to perform rapid
identification of anomalies in operation of the power system. In this
paper, we use this data to detect line outage events, discriminating
between outages on different lines.  This discrimination capability
(known in machine learning as ``classification'') is made possible by
the fact that the topological change to the grid resulting from a line
outage leads (after a transient period during which currents and
voltages fluctuate) to a new steady state of voltage and power
values. The pattern of voltage and power changes is somewhat
distinctive for different line outages. By gathering or simulating
many samples of these changes, under different load conditions, we can
train a machine-learning classifier to recognize each type of line
outage.  \modify{Further, given that it is not common in current
  practice to install PMUs on all buses, we extend our methodology to
  place a limited number of PMUs in the network in a way that
  maximizes the performance of outage detection, or to find optimal
  locations for {\em additional} PMUs in a network that is already
  instrumented with some PMUs.}

Earlier works on classification of line outages from PMU data are
based on a linear (DC) power flow model \cite{ZhuG12,CheLW14}, or make
use only of phasor angle changes \cite{TatO08,TatO09,AbdME12}, or
design a classifier that depends only linearly on the differences in
sensor readings before and after an event \cite{KimW16a}. These
approaches fail to exploit fully the modeling capabilities provided by
the AC power flow equations, the information supplied by PMUs, and the
power of modern machine learning techniques. Neural networks have the
potential to extract automatically from the observations information
that is crucial to distinguishing between outage events, transforming
the raw data vectors into a form that makes the classification more
accurate and reliable. Although the computational burden of training a
neural-network classifier is heavy, this processing can be done
``offline.'' The cost of deploying the trained classifier is
low. Outages can be detected and classified quickly, in real time,
possibly leading to faster remedial action on the grid, and less
damage to the infrastructure and to customers.  \modify{The idea of
  using neural networks on PMU data is also studied in \cite{ZhaCP16b}
  to detect multiple simultaneous line outages, in the case that PMU
  data from all buses are available along with data for power
  injections at all buses.}


The use of neural networks in deep learning is currently the subject
of much investigation. Neural networks have yielded significant
advances in computer vision and speech recognition, often
outperforming human experts, especially when the hidden information in
the raw input is not captured well by linear models.  The limitations
of linear models can sometimes be overcome by means of laborious
feature engineering, which requires expert domain knowledge, but this
process may nevertheless miss vital information hidden in the data
that is not discernible even by an expert.
We show below that, in this application to outage detection on power
grids, even generic neural network models are effective at classifying
outages accurately across wide ranges of demands and seasonal effects.
\modify{Previous works of data-based methods for outage detection only
  demonstrated outage-detecting ability of these models for a limited
  range of demand profiles. We show that neural network models can
  cope with a wider range of realistic demand scenarios, that
  incorporate seasonal, diurnal, and random fluctuations.  Although
  not explored in this paper, our methodology could incorporate
  various scenarios for power supply at generation nodes as well.}
\modify{We show too that effective outage detection can be achieved
  with information from PMUs at a limited set of network locations,
  and provide methodology for choosing these locations so as to
  maximize the outage detection performance.}

Our approach differs from most approaches to machine learning
classification in one important respect. Usually, the data used to
train a classifier is {\em historical} or {\em streaming}, gathered by
passive observation of the system under study.
Here, instead, we are able to {\em generate} the data as required, via
a high-fidelity model based on the AC power flow equations. Since we
can generate enough instances of each type of line outage to make them
clearly recognizable and distinguishable, we have an important
advantage over traditional machine learning. The role of machine
learning is thus slightly different from the usual setting. The
classifier serves as a proxy for the physical model (the AC power flow
equations), treating the model as a black box and performing the
classification task phenomenologically based on its responses to the
``stimuli'' of line outages. Though the offline computational cost of
training the model to classify outages is high, the neural network
proxy can be deployed rapidly, requiring much less online computation
than an inversion of the original model.

This work is an extension and generalization of \cite{KimW16a}, where
a linear machine learning model (multiclass logistic regression, or
MLR) is used to predict the relation between the PMU readings and the
outage event.  The neural-network scheme has MLR as its final layer,
but the network contains additional ``hidden layers'' that perform
nonlinear transformations of the raw data vectors of PMU readings. We
show empirically that the neural network gives superior classification
performance to MLR in a setting in which the electricity demands vary
over a wider range than that considered in
\cite{KimW16a}. \modify{(The wider range of demands casues the PMU
  signatures of each outage to be more widely dispersed, and thus
  harder to classify.)}
A similar approach to outage detection was discussed in
\cite{GarC16a}, using a linear MLR model, with PMU data gathered
during the transient immediately after the outage has occurred, rather
than the difference between the steady states before and after the
outage, as in \cite{KimW16a}.  Data is required from all buses in
\cite{GarC16a}, whereas in \cite{KimW16a} and in the present paper, we
consider too the situation in which data is available from only a
subset of PMUs.

\modify{Another line of work that uses neural networks for outage
  detection is reported in \cite{ZhaCP16b} (later expanded into the
  report \cite{ZhaCP17a}, which appeared after the original version of
  this paper was submitted). The neural networks used in
  \cite{ZhaCP16b,ZhaCP17a} and in our paper are similar in having a
  single hidden layer. However, the data used as inputs to the neural
  networks differs. We use the voltage angles and magnitudes reported
  by PMUs, whereas \cite{ZhaCP16b,ZhaCP17a} use only voltage angles
  along with power injection data at all buses. Moreover,
  \cite{ZhaCP16b,ZhaCP17a} require PMU data from {\em all} buses,
  whereas we focus on identifying a subset of PMU locations that
  optimizes classification performance. A third difference is that
  \cite{ZhaCP16b,ZhaCP17a} aim to detect multiple, simultaneous line
  outages using a multilabel classification formulation, while we aim
  to identify only single- or simultaneous double-line outages. The
  latter are typically the first events to occur in a large-scale grid
  failure, and rapid detection enables remedial action to be taken. We
  note too that PMU data is simulated in \cite{ZhaCP16b,ZhaCP17a} by
  using a DC power flow model, rather than our AC model, and that a
  variety of power injections are obtained in the PMU
  data not by varying over a plausible range of seasonal and diurnal
  demand/generation variations (as we do) but rather by perturbing voltage angles
  randomly and inferring the effects of these perturbations on power
  readings at the buses.}

This paper is organized as follows. In Section~\ref{sec:sparse}, we
give the mathematical formulation of the neural network model, and the
regularized formulation that can be used to determine optimal PMU
placement. We then discuss efficient optimization algorithms for
training the models in Section~\ref{sec:alg}. Computational
experiments are described Section~\ref{sec:exp}.  \arxiv{A convergence
  proof for the optimization method is presented in the Appendix.}

\section{Neural Network and Sparse Modeling}
\label{sec:sparse}
In this section, we discuss our approach of using neural network
models to identify line outage events from PMU change data, and extend
the formulation to find optimal placements of PMUs in the network.
(We avoid a detailed discussion of the AC power flow model.)
We use the following notation for outage event.
\begin{itemize}[leftmargin=*]
\parskip 0pt
\item $y_i$ denotes the outage represented by event $i$. It takes a
  value in the set $\{1,\dotsc,K\}$, where $K$ represents the total
  number of possible outage events (roughly equal to the number of
  lines in the network that are susceptible to failure).
\item $\bx_i \in \R^d$ is the vector of differences between the
  pre-outage and post-outage steady-state PMU readings,
\end{itemize}
In the parlance of machine learning, $y_i$ is known as a {\em label}
and $\bx_i$ is a {\em feature vector}. Each $i$ indexes a single item
of data; we use $n$ to denote the total number of items, which is a
measure of the size of the data set.

\subsection{Neural Network}

A neural network is a machine learning model that transforms the data
vectors $\bx_i$ via a series of transformations (typically linear
transformations alternating with simple component-wise nonlinear
transformations) into another vector to which a standard linear
classification operation such as MLR is applied. The transformations
can be represented as a network. The nodes in each layer of this
network correspond to elements of an intermediate data vector;
nonlinear transformations are performed on each of these elements. The
arcs between layers correspond to linear transformations, with the
weights on each arc representing an element of the matrix that
describes the linear transformation.  The bottom layer of nodes
contains the elements of the raw data vector while a ``softmax''
operation applied to the outputs of the top layer indicates the
probabilities of the vector belonging to each of the $K$ possible
classes. The layers / nodes strictly between the top and bottom layers
are called ``hidden layers'' and ``hidden nodes.''

A neural network is \emph{trained} by determining values of the
parameters representing the linear and nonlinear transformations such
that the network performs well in classifying the data objects
$(\bx_i,y_i)$, $i=1,2,\dotsc,n$. More specifically, we would like the
probability assigned to node $y_i$ for input vector $\bx_i$ to be
close to $1$, for each $i=1,2,\dotsc,n$.  The linear transformations
between layers are learned from the data, allowing complex
interactions between individual features to be captured. Although deep
learning lacks a satisfying theory, the layered structure of the
network is thought to mimic gradual refinement of the information, for
highly complicated tasks.  
\canomit{For example, in image classification,
successive layers aggregate the information from pixels into lines and
shapes, and eventually a recognizable object.} In our current
application, we expect the relations between the input features ---
the PMU changes before / after an outage event --- to be related to
the event in complex ways, making the choice of a neural network model
reasonable.

Training of the neural network can be formulated as an optimization
problem as follows. Let $N$ be the number of hidden layers in the
network, with $d_1,d_2,\dotsc, d_N\geq 0$ being the number of hidden
nodes in each hidden layer. ($d_0=d$ denotes the dimension of the raw
input vectors\canomit{ $\bx_i$, $i=1,2,\dotsc,n$}, while $d_{N+1}=K$ is the
number of classes.) We denote by $W_j$ the matrix of dimensions
$d_{j-1} \times d_j$ that represents the linear transformation of
output of layer $j-1$ to the input of layer $j$. The nonlinear
transformation that occurs within each layer is represented by the
function $\sigma$. With some flexibility of notation, we obtain
$\sigma(\bx)$ by applying the same transformation to each component of
$\bx$. In our model, we use the $\tanh$ function, which transforms
each element $\nu \in \R$ as follows:
\begin{equation} \label{eq:tanh}
\nu \to (e^\nu - e^{-\nu})/(e^{\nu} + e^{-\nu}).
\end{equation}
(Other common choices of $\sigma$ include the sigmoid function $\nu
\to 1/(1+e^{-\nu})$ and the rectified linear unit $\nu \to
\max(0,\nu)$.)  This nonlinear transformation is not applied at the
output layer $N+1$; the outputs of this layer are obtained by
applying an MLR classifier to the outputs of layer $N$.

Using this notation, together with $[n] = \{1,2,\dotsc,n\}$ and $[N] =
\{ 1,2,\dotsc,N\}$, we formulate the training problem as:
\begin{equation} \label{eq:nn}
 \min\nolimits_{W_1,W_2,\dotsc,W_{N+1}} \,  f(W_1,
W_2, \dotsc, W_{N+1}),
\end{equation}
where the objective is defined by 
\begin{subequations} \label{eq:f}
\begin{alignat}{2}
\label{eq:f.1}
f(W_1,\dotsc,W_{N+1})
	& := \sum_{i=1}^n \ell(\bx_i^{N+1}, && y_i) + \frac{\epsilon}{2}
	\sum_{j=1}^{N+1} \|W_j\|_F^2, \\
\mbox{subject to} \;\; \bx_i^{N+1} &= W_{N+1}\bx_i^{N}, \; && i \in [n], \\
\label{eq:f.3}
  \bx_i^j &= \sigma(W_{j} \bx_i^{j-1}), \; && i \in [n],\, j  \in [N], \\
 \bx_i^0 &= \bx_i, \; && i \in [n],
\end{alignat}
\end{subequations}
for some given regularization parameter $\epsilon \ge 0$ and Frobenius
norm $\| \cdot \|_F$, and nonnegative convex loss function $\ell$. \footnote{We
chose a small positive value $\epsilon=10^{-8}$ for our experiments,
as a positive value is required for the convergence theory\arxiv{; see in particular
Lemma~\ref{lem:bddgrad} in the Appendix}. The computational results
were very similar for $\epsilon=0$, however.} We use the constraints
in \eqref{eq:f} to eliminate intermediate variables $\bx_i^j$,
$j=1,2,\dotsc,N+1$, so that indeed \eqref{eq:nn} is an unconstrained
optimization problem in $W_1,W_2,\dotsc,W_{N+1}$.  The loss function
$\ell$ quantifies the accuracy which with the neural network predicts
the label $y_i$ for data vector $\bx_i$. As is common, we use the MLR
loss function, which is the negative logarithm of the softmax
operation, defined by
\begin{equation} \label{eq:def.l}
	\ell(\bz,y_i) := -\log \left( \frac{e^{z_{y_i}}}{ \sum_{k=1}^K e^{z_i}} \right) =
-z_{y_i} + \log \left( \sum\nolimits_{k=1}^K e^{z_i} \right),
\end{equation}
where $\bz = (z_1,z_2,\dotsc,z_K)^T$.  Since for a transformed data
vector $\bz$, the neural network assigns a probability proportional to
$\exp(z_k)$ for each outcome $k=1,2,\dotsc,K$, this function is
minimized when the neural network assigns zero probabilities to the
incorrect labels $k \ne y_i$.

In practice, we add ``bias'' terms at each layer, so that the
transformations actually have the form
\[
\bx_i^{j-1} \to W_j \bx_i^{j-1} + w_j,
\]
for some parameter $w_j \in \R^{d_j}$. We omit this detail from our
description, for simplicity of notation.

Despite the convexity of the loss function $\ell$ as a function of its
arguments, the overall objective \eqref{eq:f} is generally nonconvex
as a function of $W_1,W_2,\dotsc,W_{N+1}$, because of the nonlinear
transformations $\sigma$ in \eqref{eq:f.3}, defined by
\eqref{eq:tanh}.

\subsection{Inducing Sparsity via Group-LASSO Regularization}

In current practice, PMU sensors are attached to only a subset of
transmission lines, typically near buses. We can modify the
formulation of neural network training to determine which PMU
locations are most important in detecting line outages. Following
\cite{KimW16a}, we do so with the help of a nonsmooth term in the
objective that penalizes the use of each individual sensor, thus
allowing the selection of only those sensors which are most important
in minimizing the training loss function \eqref{eq:f}. This penalty
takes the form of the sum of Frobenius norms on submatrices of $W_1$,
where each submatrix corresponds to a particular sensor. Suppose that
$G_s \subset \{1,2,\dotsc,d\}$ is the subset of features in $\bx_i$
that are obtained from sensor $s$. If the columns $j \in G_s$ of the
matrix $W_1$ are zero, then these entries of $\bx_i$ are ignored ---
the products $W_1 \bx_i$ will be independent of the values $(\bx_i)_j$
for $j \in G_s$ --- so the sensor $s$ is not needed.  Denoting by $I$
a set of sensors, we define the regularization term as follows:
\begin{subequations} \label{eq:def.cr}
\begin{align}
\label{eq:def.c}
	c(W_1,I) &:= \sum\nolimits_{s \in I} r(W_1,G_s),
	\text{ where }\\
\label{eq:def.r}
	r(W_1,G_s) &:= \sqrt{\sum\nolimits_{i=1}^{d_1} \sum_{j \in G_s} (W_1)_{i,j}^2} =
\left\| (W_1)_{\cdot G_s} \right\|.
\end{align}
\end{subequations}
(We can take $I$ to be the full set of sensors or some subset, as
discussed in Subsection~\ref{sec:pmu_selection}.)  This form of
regularizer is sometimes known as a group-LASSO
\cite{malioutov2005sparse,MeiVB08a,WriNF09a}. With this regularization
term, the objective in \eqref{eq:nn} is replaced by
\begin{equation} \label{eq:reg-loss}
 L_I(W) := f(W_1, \dotsc, W_N) + \tau c(W_1, I),
\end{equation}
for some tunable parameter $\tau \geq 0$. A larger $\tau$ induces more
zero groups (indicating fewer sensors) while a smaller value of
$\tau$ tends to give lower training error at the cost of using more
sensors.  Note that no regularization is required on $W_i$ for $i>1$,
since $W_1$ is the only matrix that operates directly on the vectors
of data from the sensors.

We give further details on the use of this regularization in
choosing PMU locations in Subsection~\ref{sec:pmu_selection} below.
Once the desired subset has been selected, we drop the regularization
term and solve a version of \eqref{eq:nn} in which the columns of
$W_1$ corresponding to the sensors not selected are fixed at zero.


\section{Optimization and Selection Algorithms}
\label{sec:alg}

Here we discuss the choice of optimization algorithms for solving the
training problem \eqref{eq:nn} and its regularized version
\eqref{eq:reg-loss}. We also discuss strategies that use the
regularized formulation to select PMU locations, \modify{when we are
  only allowed to install PMUs on a pre-specified number of buses}.

\subsection{Optimization Frameworks}
\label{sec:opt.frameworks}

\begin{algorithm}
	\caption{Greedy heuristic for feature selection}
	\label{alg:greedy}
	Given $\epsilon,\tau > 0$, $\text{\#max\_group} \in \N$, set $I$ of possible sensor locations, and disjoint groups $\{G_s\}$ such that
	$\bigcup_{s \in I} G_s \subset \{1,\dots, d\}$\;\
	Set $G \leftarrow \emptyset$; \\
	\For{$k=1,\dotsc, $\textit{\#max\_group}}{
		\eIf{$k > 1$}{
			Let the initial point be the solution from the previous
			iteration\;}{
			Randomly initialize  $W_i \in \R^{d_{i-1} \times d_{i}}$,
			$i \in [N+1]$\;}
		Approximately solve \eqref{eq:reg-loss} with the given $\tau$ and the
		current $I$ by SpaRSA\;
		$\tilde{s} := \arg\max_{s\in I}\quad r(W_1, G_s)$\;
		\If{$r(W_1, G_{\tilde{s}}) = 0$}{
			Break\;}
		$I \leftarrow I \setminus \tilde{s}$, $G \leftarrow G \cup
		\{\tilde{s} \}$\;
	}
	Output $G$ as the selected buses and terminate\;
\end{algorithm}
We solve the problem \eqref{eq:nn} with the popular L-BFGS algorithm
\cite{LiuN89a}.  Other algorithms for smooth nonlinear optimization
can also be applied; we choose L-BFGS because it requires only
function values and gradients of the objective, and because it has
been shown in \cite{NgiCL11a} to be efficient for solving neural
network problems. To deal with the nonconvexity of the objective, we
made slight changes of the original L-BFGS, following an
idea in \cite{LiF01a}. Denoting by $\bs_t$ the difference between the
iterates at iterations $t$ and $t+1$, and by $\by_t$ the difference
between the gradients at these two iterations, the pair 
$(\bs_t,\by_t)$ is not used in computing subsequent search
directions if $\bs_t^T \by_t \ll \bs_t^T \bs_t$. This strategy ensures
that the Hessian approximation remains positive definite, so the
search directions generated by L-BFGS will be descent directions.

We solve the group-regularized problem \eqref{eq:reg-loss} using
SpaRSA \citep{WriNF09a}, a proximal-gradient method that requires only
the gradient of $f$ and an efficient proximal solver for the
regularization term.  As shown in \cite{WriNF09a}, the proximal
problem associated with the group-LASSO regularization has a closed
form solution that is inexpensive to compute.

In the next section, we discuss details of two bus selection
approaches, and how to compute the gradient of $f$ efficiently.

\subsection{Two Approaches for PMU Location}
\label{sec:pmu_selection}

We follow \cite{KimW16a} in proposing two approaches for selecting PMU
locations.  In the first approach, we set $I$ in \eqref{eq:reg-loss}
to be the full set of potential PMU locations, and try different
values of the parameter $\tau$ until we find a solution that has the
desired number of nonzero submatrices $(W_1)_{\cdot j}$ for $j \in I$,
which indicate the chosen PMU locations.

The second approach is referred to as the ``greedy heuristic'' in
\cite{KimW16a}.  We initialize $I$ to be the set of candidate
locations for PMUs. (We can exclude from this set locations that are
already instrumented with PMUs and those that are not to be considered
as possible PMU locations.) We then minimize \eqref{eq:reg-loss} with this
$I$, and select the index $s$ that satisfies
\begin{equation*}
	s = \arg\max\nolimits_{s \in I}\quad r(W_1, G_s)
\end{equation*}
as the next PMU location. This $s$ is removed from $I$, and we
minimize \eqref{eq:reg-loss} with the reduced $I$.  This process is
repeated until the required number of locations has been selected.
The process is summarized in Algorithm \ref{alg:greedy}.

\subsection{Computing the Gradient of the Loss Function}

In both SpaRSA and the modified L-BFGS algorithm, the gradient and the
function value of $f$ defined in \eqref{eq:f} are needed at every
iteration.  We show how to compute these two values efficiently given
any iterate $W=(W_1,W_2,\dotsc,W_{N+1})$.  Function values are
computed exactly as suggested by the constraints in \eqref{eq:f}, by
evaluating the intermediate quantities $\bx_i^j$, $j \in [N+1]$, $i
\in [n]$ by these formulas, then finally the summation in
\eqref{eq:f.1}. The gradient involves an adjoint calculation. By
applying the chain rule to the constraints in \eqref{eq:f}, treating
$\bx_i^j$, $j \in [N+1]$, as variables alongside
$W_1,W_2,\dotsc,W_{N+1}$, we obtain
\begin{subequations} \label{eq:grads}
\begin{align}
	\nabla_{W_{N+1}}f &= \sum_{i=1}^n
	\nabla_{\bx_i^{N+1}}\ell(\bx_i^{N+1}, y_i) (\bx_i^N)^T
	+ \epsilon W_{N+1},
	\label{eq:WN}\\
	\nabla_{\bx^N_i} f &= \nabla_{\bx^{N+1}_i}
		\ell(\bx_i^{N+1}, y_i) W_{N+1}^T,
		\label{eq:dummy} \\
	\nabla_{\bx_i^j} f &=\nabla_{\bx^{j+1}_i} f\cdot \sigma'(W_{j+1}\bx^j_i)W_{j+1}^T , \label{eq:xi} \\
& \qquad j=N-1,\dotsc,0, \nonumber
		\\
		\nabla_{W_{j}} f &= \sum_{i=1}^n
		\nabla_{\bx^j_i} f\cdot
		\sigma'(W_{j}\bx^{j-1}_i)(\bx_i^{j-1})^T
		+ \epsilon W_{j},\label{eq:Wi} \\
		&\qquad j=1,\dotsc,N. \nonumber
\end{align}
\end{subequations}
Since $\sigma$ is a pointwise operator that maps $\R^{d_i}$ to
$\R^{d_i}$, $\sigma'(\cdot)$ is a diagonal matrix such that
$\sigma'(\bz)_{i,i} = \sigma'(z_i)$.  The quantities $\sigma'()$ and
$\bx_i^j$, $j=1,2,\dotsc,N+1$ are computed and stored during the
calculation of the objective. Then, from \eqref{eq:dummy} and
\eqref{eq:xi}, the quantities $\nabla_{\bx^j_i} f$ from $j=N,N-1,
\dotsc,0$ can be computed in a reverse recursion.  Finally, the
formulas \eqref{eq:Wi} and \eqref{eq:WN} can be used to compute the
required derivatives $\nabla_{W_j} f$, $j=1,2,\dotsc,N+1$.

\subsection{Training and Validation Procedure}

In accordance with usual practice in statistical analysis involving
regularization parameters, we divide the available data into a {\em
  training set} and a {\em validation set}. The training set is a
randomly selected subset of the available data --- the pairs
$(\bx_i,y_i)$, $i=1,2,\dotsc,n$ in the notation above --- that is used
to form the objective function whose solution yields the parameters
$W_1,W_2,\dotsc,W_{N+1}$ in the neural network. The validation set
consists of further pairs $(\bx_i,y_i)$ that aid in the choice of the
regularization parameter, which in our case is the parameter $\tau$ in
the greedy heuristic procedure of Algorithm~\ref{alg:greedy},
\modify{described in Sections~\ref{sec:opt.frameworks} and
  \ref{sec:pmu_selection}.}
We apply the greedy heuristic for $\tau \in
\{2^{-8},2^{-7},\dotsc,2^7,2^8\}$ and deem the optimal value to be the
one that achieves the most accurate outage identification on the
validation set.  We select initial points for the training randomly,
so different solutions $W_1,W_2,\dotsc,W_{N+1}$ may be obtained even
for a single value of $\tau$.  To obtain a ``score'' for each value of
$\tau$, we choose the best result from ten random starts. The final
model is then obtained by solving \eqref{eq:nn} over the buses
selected on the best of the ten validation runs, that is, fixing the
elements of $W_1$ that correspond to non-selected buses at zero.


Note that validation is not needed to choose the value of $\tau$ when
we solve the regularized problem \eqref{eq:reg-loss} directly, because
in this procedure, we adjust $\tau$ until a predetermined number of
buses is selected.


There is also a {\em testing set} of pairs $(\bx_i,y_i)$. This is data
that is used to evaluate the bus selections produced by the procedures
above. In each case, the tuned models obtained on the selected buses
are evaluated on the testing set.


\section{Experiments} \label{sec:exp}

We perform simulations based on grids from the IEEE test set archive
\citep{IEEE14a}. Many of our studies focus on the IEEE-57bus
case. Simulations of grid response to varying demand and outage
conditions are performed using MATPOWER \citep{ZimM11a}.  \modify{We
  first show that high accuracy can be achieved easily when PMU
  readings from all buses are used. We then focus on the more
  realistic (but more difficult) case in which data from only a limited
  number of PMUs is used. In both cases, we simulate PMU readings over
  a wide range of power demand profiles that encompass the profiles
  that would be seen in practice over different seasons and at
  different times of day.}

\subsection{Data Generation}
\label{subsec:data}
We use the following procedure from \cite{KimW16a} to generate the
data points using a stochastic process and MATPOWER.
\begin{enumerate}[leftmargin=*]
\item We consider the full grid defined in the IEEE specification, and
  also the modified grid obtained by removing each transmission line
  in turn.
\item For each demand node, define a baseline demand value from the
  IEEE test set archive as the average of the load demand over 24
  hours.
\item \modify{To simulate different ``demand averages'' for different seasons,
  we scale the baseline demand value for each node by the values in
  $\{0.5,0.75,1,1.25,1.5\}$, to yield five different baseline demand
  averages for each node.}  (Note: In \cite{KimW16a}, a narrower range
  of multipliers was used, specifically $\{0.85,1,1.15\}$, but each
  multiplier is considered as a different independent data set.)
\item Simulate a 24-hour fluctuation in demand by an adaptive
  Ornstein-Uhlenbeck process as suggested in \cite{PerK11a},
  independently and separately on each demand bus.
\item \modify{This fluctuation is overlaid on the demand average for each bus
	to generate a 24-hour load demand profile.}
  \item \modify{Obtain training, validation, and test points from these
    24-hour demand profiles for each node by selecting different
timepoints from this 24-hour period, as described below.}
\item If any combination of line outage and demand profile yields a
  system for which MATPOWER cannot identify a feasible solution for
  the AC power flow equations, we do not add this point to the data
  set. Lines connecting the same pair of buses are considered as a
  single line; we take them to be all disconnected or all connected.
\end{enumerate}
This procedure was used to generate training, validation,
and test data. In each category, we generated equal numbers of training
points for each feasible case in each of the five scale factors $\{0.5,0.75,1,1.25,1.5\}$.
For each feasible topology and each combination of parameters above,
we generate $20$ training points from the first 12 hours of the
24-hour simulation period, and $10$ validation points and $50$ test
points from the second 12-hour period.
Summary information about the IEEE power systems we use in the
experiments with single line outage is shown in Table~\ref{tbl:data}.
The column ``Feas.'' shows the number of lines whose removal still
result in a feasible topology for at least one scale factor, while the
number of lines whose removal result in infeasible topologies for all
scale factors or are duplicated is indicated in the column
``Infeas./Dup.''  The next three columns show the number of data
points in the training / validation / test sets. As an example: The
number of training points for the 14-Bus case (which is 1840) is
approximately 19 (number of feasible line removals) times 5 (number of
demand scalings) times 20 (number of training points per
configurations). The difference between this calculated value of
$1900$ and the $1840$ actually used is from that the numbers of
feasible lines under different scaling factors are not identical, and
higher scaling factors resulted in more infeasible cases.  The last
column in Table~\ref{tbl:data} shows the number of components in each
feature vector $\bx_i$. There are two features for each bus, being
changes in phase angle and voltage magnitude with respect to the
original grid under the same demand conditions. There are another two
additional features in all cases, one indicating the power generation
level (expressed as a fraction of the long-term average), and the
other one indicating a bias term manually added to the data.
\begin{table}
	\caption{The systems used in our experiment and statistics of the synthetic
	data.}
	\label{tbl:data}
	\centering
	\begin{tabular}{@{}l|r|r|r|r|r|r@{}}
		System & \multicolumn{2}{c|}{\#lines} & \#Train & \#Val &
		\#Test & \#Features\\
		\hline
		& Feas. & Infeas./Dup. & & & &\\
		\hline
		14-Bus & 19 & 1 & 1,840 & 920 & 4,600 & 30\\
		\hline
		30-Bus & 38 & 3 & 3,680 & 1,840 & 9,200 & 62 \\
		\hline
		57-Bus & 75 & 5 & 5,340 & 2,670 & 13,350 & 116 \\
		\hline
		118-Bus & 170 & 16 & 16,980 & 8,490 & 42,450 & 238\\
	\end{tabular}
\end{table}

\subsection{Neural Network Design} \label{sec:nn.design}


Configuration and design of the neural network is critical to
performance in many applications. In most of our
experiments, we opt for a simple design in which there is just a
single hidden layer: $N=1$ in the notation of \eqref{eq:nn}. We assume
that the matrices $W_1$ and $W_2$ are dense, that is, all nodes in any one
layer are connected to all nodes in adjacent layers. It remains to
decide how many nodes $d_1$ should be in the hidden layer. Larger values
of $d_1$ lead to larger matrices $W_1$ and $W_2$ and thus more
parameters to be chosen in the training process. However, larger $d_1$
can raise the possibility of overfitting the training
data, producing solutions that perform poorly on the other, similar
data in the validation and test sets.

We did an experiment to indicate whether overfitting could be an issue
in this application. We set $d_1=200$, and solved the unregularized
training problem \eqref{eq:nn} using the modified L-BFGS algorithm
with $50,000$ iterations.  Figure~\ref{fig:light} represents the
output of each of the 200 nodes in the hidden layer for each of the
$13,350$ test examples. Since the output is a result of the $\tanh$
transformation \eqref{eq:tanh} of the input, it lies in the range
$[-1,1]$. We color-code the outputs on a spectrum from red to blue,
with red representing $1$ and blue representing $-1$. A significant
number of columns are either solid red or solid blue. The hidden-layer
nodes that correspond to these columns play essentially no role in
distinguishing between different outages; similar results would be
obtained if they were simply omitted from the network. The presence of
these nodes indicates that the training process avoids using all $d_1$
nodes in the hidden layer, if fewer than $d_1$ nodes suffice to attain a
good value of the training objective. Note that overfitting is avoided
at least partially because we stop the training procedure with a
rather small number of iterations, which can be viewed as another type
of regularization \citep{CarLG01a}.

In our experiments, we used $d_1=200$ for the larger grids (57 and 114
buses) and $d_1=100$ for the smaller grids (14 and 30 buses).  The
maximum number of L-BFGS iterations for all neural networks is set to
$50,000$, while for MLR models we terminate it either when the number of iterations reaches $500,000$ or when the gradient is
smaller than a pre-specified value ($10^{-3}$ in our experiments), as
linear models do not suffer much from overfitting.

\begin{figure}
\centering
\includegraphics[width=\linewidth,height=.7in]{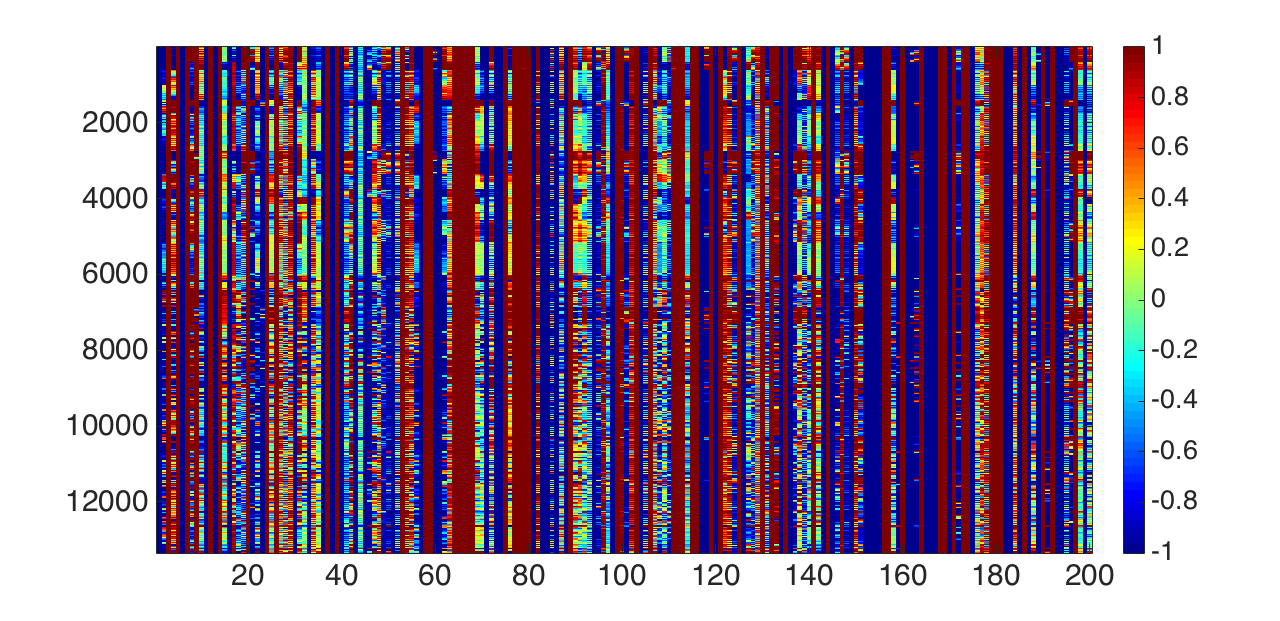}
\caption{Output of the hidden layer nodes of a one-layer neural
  network with $200$ hidden nodes applied to the problem of detecting
  line outages on the IEEE 57-bus grid.  Columns with a single color
  (dark red or dark blue) indicate nodes that output
  the same value regardless of the feature vector $\bx_i$ that was
  input into the neural network. Such nodes play little or no
  role in discriminating between different line outages.}
	\label{fig:light}
\end{figure}

\subsection{Results on All Buses}

We first compare the results between linear multinomial logistic
regression (MLR) (as considered in \cite{KimW16a}) and a fully
connected neural network with one hidden layer, where the PMUs are
placed on all buses. Because we use all the buses, no validation phase
is needed, because the parameter $\tau$ does not appear in the model.
Table~\ref{tbl:full} shows error rates on the testing set.
We see that in the difficult cases, when the linear model has error
rates higher than $1\%$, the neural network obtains markedly better
testing error rates.

\begin{table}
\begin{center}
\caption{PMUs on all buses: Test error rates for single-line outage.}
\label{tbl:full}
\begin{tabular}{l|r|r|r|r}
Buses& 14 & 30 & 57 & 118 \\
\hline
Linear MLR & 0.00\% & 1.76\% & 4.50\% & 15.19\% \\
Neural network& 0.43\% & 0.03\% & 0.91\% & 2.28\%
\end{tabular}
\end{center}
\end{table}

\subsection{Results on Subset of Buses}

We now focus on the 57-bus case, and apply the greedy heuristic
(Algorithm~\ref{alg:greedy}) to select a subset of buses for PMU
placement, for the neural network with one hidden layer of 200 nodes.
We aim to select 10 locations. Figure~\ref{fig:group} shows the
locations selected at each run. Values of $\tau$ used were $\{2^{-8},
2^{-7}, \dotsc, 2^8\}$, with ten runs performed for each value of
$\tau$. On some runs, the initial point is close to a bad local
optimum (or saddle point) and the optimization procedure terminates
early with fewer than $10 \times 2$ columns of non-zeros in $W_1$
(indicating that fewer than 10 buses were selected, as each bus
corresponds to 2 columns). The resulting models have poor performance,
and we do not include them in the figure.

Even though the random initial points are different on each run, the
groups selected for a fixed $\tau$ tend to be similar on all runs when
$\tau \le 2$.  For larger values of $\tau$, including the value $\tau
= 2^4$ which gives the best selection performance, the locations
selected on different runs are often different.
(For the largest values of $\tau$, fewer than 10 buses are selected.)

Table~\ref{tbl:sparse} shows testing accuracy for the ten PMU
locations selected by both the greedy heuristic and regularized
optimization with a single well-chosen value of $\tau$. Both the
neural network and the linear MLR classifiers were tried. The groups of
selected buses are shown for each case. These differ significantly; we
chose the ``optimal'' group from among these to be the one with the
best validation score.  \canomit{The L-BFGS optimization procedure was run on
the selected group of buses for 50,000 iterations on the neural
network model, and 500,000 on the linear case (or terminated when the
gradient is too small).}  We note the very specific choice of $\tau$
for linear MLR (group-LASSO). In this case, the number of groups
selected is extremely sensitive to $\tau$. In a very small range
around $\tau=14.4898999$, the number of buses selected varies between
8 and 12.
\smodify{We report two types of error rates here. In the column
  ``Err. (top1)'' we report the rate at which the outage that was
  assigned the highest probability by the classifier was not the outage
  that actually occurred. In ``Err. (top2)'' we score an error only if
  the true outage was not assigned either the highest or the second-highest
  probability by the classifier.  We note here that ``top1'' error
  rates are much higher than when PMU data from all buses is used,
  although that the neural network yields significantly better results
  than the linear classifier. However, ``top2'' results are excellent
  for the neural network when the greedy heuristic is used to select
  bus location.}

\begin{table*}
	\centering
	\caption{Comparison of different approaches for selecting $10$ buses
	on the IEEE 57-bus case, after $50,000$ iterations for neural networks and $500,000$ iterations for linear MLR models.}
	\label{tbl:sparse}
	\begin{tabular}{l|r|r|r|r}
		Model & $\tau$ & Buses selected & Err. (top1)  & \modify{Err.
		(top2) }\\
		\hline
		Linear MLR (greedy) & $2$ & [5 16 20 31 40 43 44 51 53 57] &
		29.7\% &\modify{ 8.4\%}\\
		Neural Network (greedy) & $16$ & [5 20 31 40 43 50 51 53 54
		57] & 7.1\% & \modify{0.1\% }\\
		\hline
		Linear MLR (group-LASSO) & $14.4898999$ & [2 4 5 6 7 8 18 27 28 29] &
		54.4\% & \modify{39.4\% }\\
		Neural Network (group-LASSO) & $48$ & [4 5 6 7 8 18 26
		27 28 55] & 24.1\% & \modify{12.9\% }\\
	\end{tabular}
\end{table*}

\begin{table*}
	\centering
	\caption{Comparison of different approaches for selecting $14$ buses
	on the IEEE 57-bus case,  after $50,000$ iterations for neural networks and $500,000$ iterations for linear MLR models.}
	\label{tbl:sparse14}
	\begin{tabular}{l|r|r|r|r}
	Model & $\tau$ & Buses selected & Err. (top1) & \modify{Err. (top2)} \\
		\hline
		Linear MLR (greedy) & $2$ & [5 16 17 20 26 31 39 40 43 44 51
		53 54 57] &
		21.8\% &\modify{ 3.8\%}\\
		Neural Network (greedy) & $16$ & [5 6 16 24 27 31 39 40 42 50
		51 52 53 54] & 5.2\% &\modify{ 0.3\% }\\
		\hline
		Linear MLR (group-LASSO) & $13$ & [2 4 5 7 8 17 18 27 28 29 31
		32 33 34] &
		42.1\% & \modify{25.3\%}\\
		Neural Network (group-LASSO) & $44$ & [4 7 8 18 24 25 26 27 28
		31 32 33 39 40] & 6.2\% & \modify{0.6\% }\\
	\end{tabular}
\end{table*}

\begin{figure}
	\begin{center}
	\includegraphics[width=\linewidth,height=.5in]{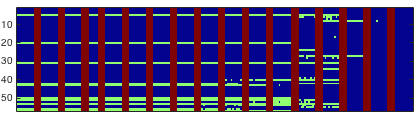}
	\caption{Groups selected on the $57$-bus case for different
          runs and different values of $\tau$ in the greedy
          heuristic applied on the neural network problem \eqref{eq:reg-loss}. Each row represents a group and each column
          represents a run. Ten runs are plotted for each value of
          $\tau$. From left to right (separated by brown vertical
          lines), these values are $\tau = 2^{-8},2^{-7},\dotsc,
          2^8$. Green indicates selected groups; dark blue are groups
          not selected.}
	\label{fig:group}
\end{center}
\end{figure}

Table~\ref{tbl:sparse14} repeats the experiment of
Table~\ref{tbl:sparse}, but for 14 selected buses rather than
10. Again, we see numerous differences between the subsets of buses
selected by the greedy and group-LASSO approaches, for both the linear
MLR and neural networks. The neural network again gives significantly
better test error rates than the linear MLR classifier, \smodify{and
  the ``top2'' results are excellent for the neural network, for both
  group-LASSO and greedy heuristics.} Possibly the most notable
difference with Table~\ref{tbl:sparse} is that the buses selected by
the group-LASSO network for the neural network gives much better
results for 14 buses than for 10 buses.  However, since it still
performs worse than the greedy heuristic, the group-LASSO approach is
not further considered in later experiments.

\subsection{Why Do Neural Network Models Achieve Better Accuracy?} 
\label{sec:why}

\begin{table*}
	\centering
	\caption{Instance distribution before and after neural network
		transformation for the IEEE 57-bus data set. In the last two
		columns, 10 buses are selected by the Greedy heuristic.}
	\label{tbl:transform2}
	\begin{tabular}{@{}l@{}|ccc@{}}
		& Full PMU Data & Selected PMUs
		& Selected PMUs, after neural network transformation\\
		\hline
		mean $\pm$ std dev. distance to centroid & $0.30 \pm 0.14$ & $0.27 \pm 0.12 $ & $2.30 \pm 1.01$ \\
		mean $\pm$ std dev. between-centroid distance & $0.17 \pm 0.14$ & $0.08 \pm 0.05$ & $3.27 \pm 1.10$
	\end{tabular}
\end{table*}

Reasons for the impressive effectiveness of neural networks in certain
applications are poorly understood, and are a major research topic 
in machine learning. For this specific problem, we compare the
distribution of the raw feature vectors with the distribution of
feature vectors obtained after transformation by the hidden layer. The
goal is to understand whether the transformed vectors are in some
sense more clearly separated and thus easier to classify than the
original data.

We start with some statistics of the clusters formed by feature
vectors of the different classes. For purposes of discussion, we
denote $\bx_i$ as
the feature vector, which could be the full set of PMU readings, the
reduced set obtained after selection of a subset of PMU locations, or
the transformed feature vector obtained as output from the hidden
layer, according to the context. For each $j\in \{1,2,\dotsc,K\}$, we gather all those feature vectors
$\bx_i$ with label $y_i=j$, and denote the centroid of this cluster by
$c_j$. We track two statistics: the mean / standard deviation of the
distance of feature vectors $\bx_i$ to their cluster centroids, that
is, $\| \bx_i - c_{y_i} \|$ for $i=1,2,\dotsc,n$; and the mean /
standard deviation of distances between cluster centroids, that is,
$\|c_j-c_k\|$ for $j,k \in \{1,2,\dotsc,K\}$. We analyze these
statistics for three cases, all based on the IEEE 57-Bus network:
first, when $\bx_i$ are vectors containing full PMU data; second, when
$\bx_i$ are vectors containing the PMU data from the 10 buses selected
by the Greedy heuristic; third, the same data vectors as in the second
case, but after they have been transformed by the hidden layer
of the neural network.

Results are shown in Table~\ref{tbl:transform2}. For the raw data
(first and second columns of the table), the distances within clusters
are typically smaller than distances between centroids. (This happens
because the feature vectors within each class are ``strung out''
rather than actually clustered, as we see below.) For the transformed
data (last column) the clusters are generally tighter and more
distinct, making them easier to distinguish.


Visualization of the effects of hidden-layer transformation is
difficult because of the high dimensionality of the feature
vectors. Nevertheless, we can gain some insight by projecting into
two-dimensional subspaces that correspond to some of the leading
principal components, which are the vectors obtained from the singular
value decomposition of the matrix of all feature vectors $\bx_i$,
$i=1,2,\dotsc,n$. Figure~\ref{fig:pca} shows two graphs. Both show
training data for the same $5$ line outages for the IEEE 57-Bus data set, with
each class coded by a particular color and shape. In both graphs, we
show data vectors obtained after 10 PMU locations were selected with
the Greedy heuristic. In the left graph, we plot the coefficients of
the first and fifth principal components of each data vector. The
``strung out'' nature of the data for each class reflects the nature
of the training data. Recall that for each outage / class, we selected
20 points from a 12-hour period of rising demand, at 5 different
scalings of overall demand level. For the right graph in
Figure~\ref{fig:pca}, we plot the coefficients of the first and third
principal components of each data vector {\em after} transformation by
the hidden layer. For both graphs, we have chosen the two principal
components to plot to be those for which the separation between
classes is most evident. For the left graph (raw data), the data for
classes 3, 4, and 5 appear in distinct regions of space, although the
border between classes 4 and 5 is thin.  For the right graph (after
transformation), classes 3, 4, and 5 are somewhat more
distinct. Classes 1 and 2 are difficult to separate in both plots,
although in the right graph, they no longer overlap with the other
three classes.  The effects of tighter clustering and cleaner
separation after transformation, which we noted in
Table~\ref{tbl:transform2}, are evident in the graphs of
Figure~\ref{fig:pca}.


\begin{figure}
	\centering
	\begin{tabular}{cc}
	\subfloat[Original feature space after bus selection (the 1st and 5th principal axes)]{
	\includegraphics[width=1.5in,height=.9in]{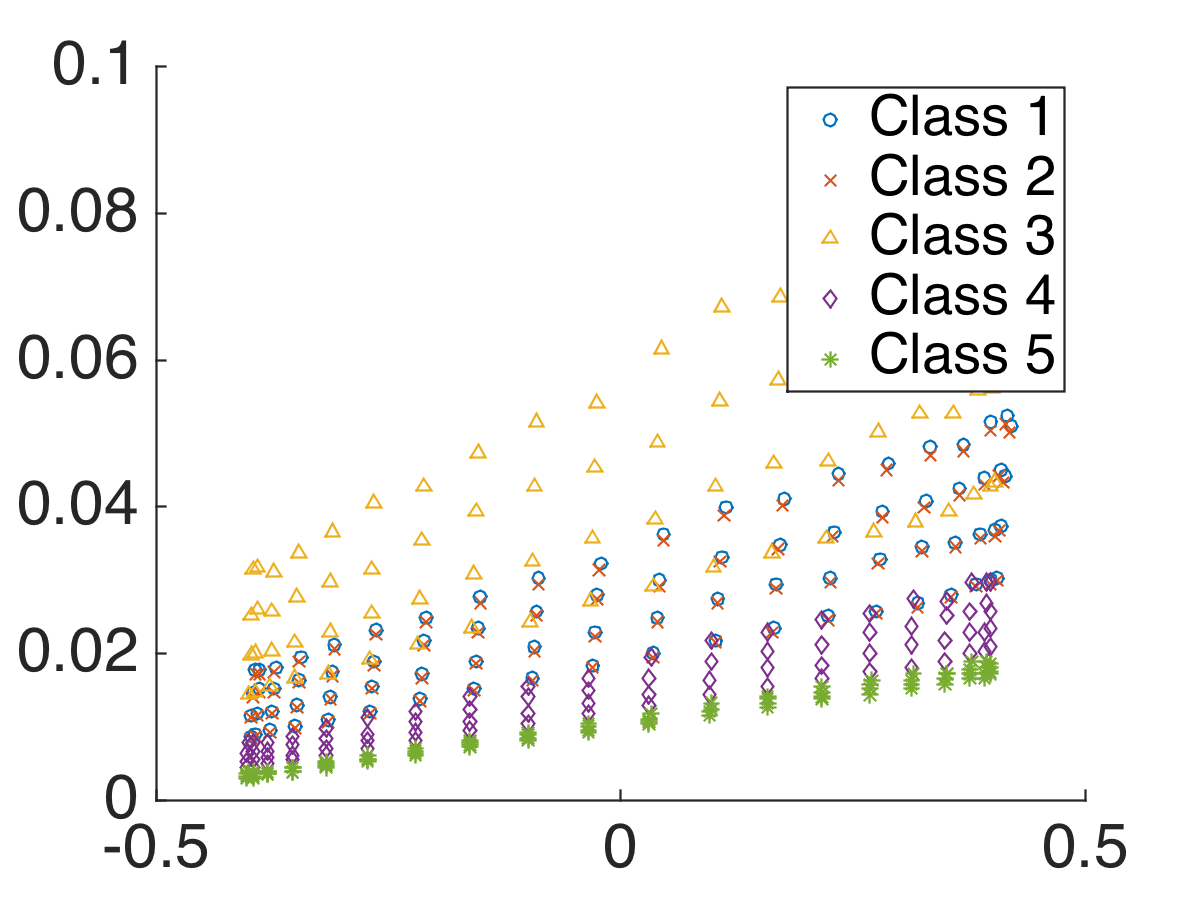}
}&
	\subfloat[The feature space after neural network transformation (the 1st and 3rd principal axes)]{
		\includegraphics[width=1.5in,height=.9in]{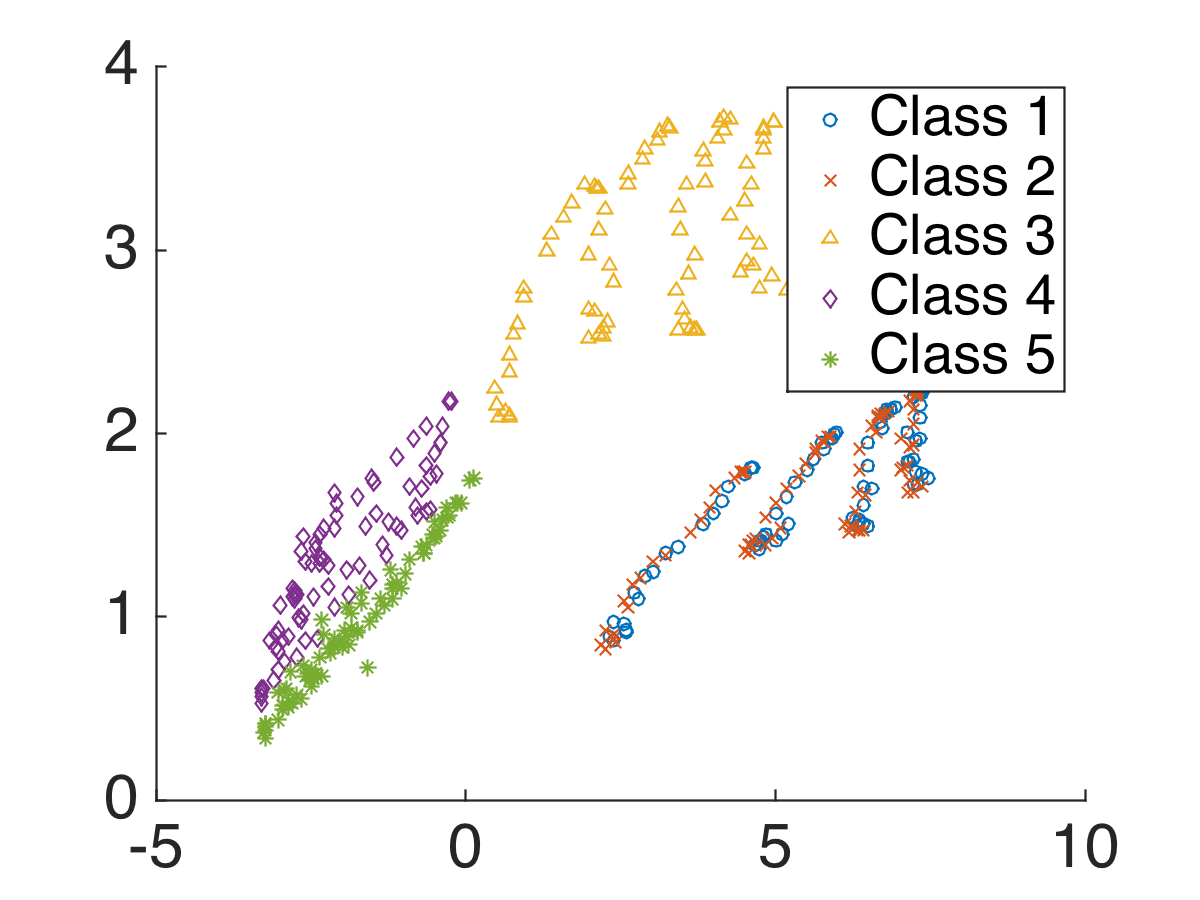}}
	\end{tabular}
	\caption{Data representation after dimension reduction to
	2D. Different colors/styles represent data points of different labels.}
	\label{fig:pca}
\end{figure}



\begin{table}
	\caption{Statistics of the synthetic data for double lines
		outage.}
	\label{tbl:doubledata}
	\centering
	\begin{tabular}{@{}l|r|r|r|r|r}
		System & \#classes & \#Train & \#Val &
		\#Test & \#Features\\
		\hline
		14-Bus & 182 & 16,420 & 8,210 & 41,050 & 30\\
		\hline
		30-Bus & 715 & 66,160 & 33,080 & 165,400 & 62\\
	\end{tabular}
\end{table}
\begin{table}
	\begin{center}
		\caption{Error rates of placing PMUs on all buses for double lines outage.}
		\label{tbl:doublefull}
		\begin{tabular}{l|r|r}
			& 14-bus & 30-bus\\
			\hline
			Linear MLR &  26.07\% & 36.32\%\\
			Neural network with one hidden layer& 0\% & 0.65\%
		\end{tabular}
	\end{center}
\end{table}
\begin{table*}
	\centering
	\caption{Comparison of different approaches for sparse PMU
		placement for double outages detection.}
	\label{tbl:doublesparse}
	\begin{tabular}{l|l|l|r|r|r|r}
		Case & Number of PMU & Model & $\tau$ & Buses selected & Err.
		(top1) & \modify{Err. (top2)}\\
		\hline
		\multirow{2}{*}{14-bus} & \multirow{2}{*}{3} &
		Linear MLR (greedy) & $8$ & [3 5 14] & 83.0\% &
		\modify{71.7\%}\\
		& & Neural Network (greedy) & $8$ & [3 12 13] & 4.3\% &
		\modify{0.9\%}\\
		\hline
		\multirow{2}{*}{30-bus} & \multirow{2}{*}{5} &
		Linear MLR (greedy) &$0.5$ & [4 5 17 23 30] & 90.6\% &
		\modify{84.5\% }\\
		& & Neural Network (greedy) & $8$ & [5 14 19 29 30] & 12.7\% &
		\modify{5.6\%}\\
	\end{tabular}
\end{table*}
\subsection{Double-Line Outage Detection}

We now extend our identification methodology to detect not just
single-line outages, but also outages on two lines simultaneously.
The number of classes that our classifier needs to distinguish between
now scales with the {\em square} of the number of lines in the grid,
rather than being approximately equal to the number of lines. For this
much larger number of classes, we generate data in the manner
described in Section~\ref{subsec:data}, again omitting cases where the
outage results in an infeasible network. Table~\ref{tbl:doubledata}
shows the number of classes for the 14- and 30-bus networks, along
with the number of training / validation / test points. Note in particular that there are 182 distinct outage events
for the 14-bus system, and 715 distinct events for the 30-bus system.

Table~\ref{tbl:doublefull} shows results of our classification
approaches for the case in which PMU observations are made at all
buses. The neural network model has a single hidden layer of 100
nodes. The neural network has dramatically better performance than the
linear MLR classifier on these problems, attaining a zero error rate
on the 14-bus tests.

We repeat the experiment using a subset of buses chosen with the
greedy heuristic described in Section~\ref{sec:pmu_selection} --- 3
buses for the 14-bus network and 5 buses for the 30-bus network. Given
the low dimensionality of the feature space and the large number of
classes, these are difficult problems.  \modify{(Because it was shown
  in the previous experiments that the group-LASSO approach has
  inferior performance to the greedy heuristic, we omit it from this
  experiment.)}  \smodify{As we see in Table~\ref{tbl:doublesparse},
  the linear MLR classifiers do not give good results, with ``top1''
  and ``top2'' error rates all in excess of 71\%.  Much better results
  are obtained for neural network with bus selection performed by the
  greedy heuristic, which obtains ``top2'' error rates of less than
  1\% in the 14-bus case and 5.6\% in the 30-bus case.}

\section{Conclusions} \label{sec:conclusion}

This work describes the use of neural networks to detect single- and
double-line outages from PMU data on a power grid. We show significant
improvements in classification performance over the linear multiclass
logistic regression methods described in \cite{KimW16a}, particularly
when data about the PMU signatures of different outage events is
gathered over a wide range of demand conditions. By adding
regularization to the model, we can determine the locations to place a
limited number of PMUs in a way that optimizes classification
performance.  Our approach uses a high-fidelity AC model of the grid
to generate data examples that are used to train the neural-network
classifier. Although (as is true in most applications of neural
networks) the training process is computationally heavy, the
predictions can be obtained with minimal computation, allowing the
model to be deployed in real time.


\bibliographystyle{IEEEtran}
\bibliography{nn_power}
\arxiv{
\clearpage

\pagebreak

\appendix

\subsection{Introduction and Implementation of the SpaRSA Algorithm}

We solve the nonsmooth regularized problem \eqref{eq:reg-loss} by
SpaRSA \cite{WriNF09a}, a proximal gradient algorithm. When applied to
\eqref{eq:reg-loss}, iteration $t$ of SpaRSA solves the following
problem, for some scalar $\alpha_t > 0$:
\begin{align} 
\nonumber
W^{t+1} & := \arg\min_{W} \frac12 \left\|W - \left(W^t - \frac{1}{\alpha_t} \nabla
f(W^t) \right) \right\|_F^2 \\
\label{eq:sparsa}
& \quad\quad\quad + \frac{\tau}{\alpha_t} c(W^t_1, I),
\end{align}
where $W^t := [W^t_1,\dotsc, W^t_{N+1}]$ denotes the $t$th iterate of
$W$. By utilizing the structure of $c(\cdot,I)$ in \eqref{eq:def.cr},
we can solve \eqref{eq:sparsa} inexpensively, in closed form.  For the value
of $\alpha_t$, at any given iteration $t$, we follow the suggestion in
\cite{WriNF09a} to start at a certain guess, and gradually increase it
until the solution of \eqref{eq:sparsa} satisfies
\begin{equation}
	f(W^{t+1}) < f(W^t) - \frac{\sigma}{2}\alpha_t \|W^{t+1} -
	W^t\|_F^2,
	\label{eq:decrease}
\end{equation}
for some small positive value of $\sigma$ (typically
$\sigma=10^{-3}$).

\subsection{Key Lemmas for Convergence Analysis}

We now analyze the convergence guarantee for SpaRSA applied to
\eqref{eq:reg-loss}. First, we establish bounds on the gradient and
Hessian of $f$.
We do not restrict using \eqref{eq:tanh} as the choice of $\sigma$. Instead, we only require that $\sigma$ is twice-continuously differentiable. 

\begin{lemma}	\label{lem:bddgrad}
Given any initial point $W^0$, there exists $c_1 \geq 0$ such that
\begin{equation} \label{eq:bddgrad}
\|\nabla f(W)\| \leq c_1
\end{equation}
in the level set $\{W\mid f(W) \leq f(W^0)\}$.
\end{lemma}
\begin{proof}
Because the loss function $\ell$ defined by \eqref{eq:def.l} is
nonnegative, we see from \eqref{eq:f} that
\[
		f(W) \geq \frac{1}{2} \epsilon \|W\|_F^2,
\]
and therefore $\{W\mid f(W) \leq f(W^0)\}$ is a subset of
\begin{equation}
		B\left(0,\sqrt{\frac{2}{\epsilon} f(W^0)}\right) := \left\{W\mid \|W\|_F \leq
			\sqrt{\frac{2}{\epsilon} f(W^0)}\right\},
			\label{eq:B}
\end{equation}
which is a compact set.  By the assumption on $\sigma$ and $\ell$,
$\|\nabla f(W)\|$ is a continuous function with respect to $W$.
Therefore, we can find $c_1 \geq 0$ such that \eqref{eq:bddgrad} holds
within the set \eqref{eq:B}.  Since the level set is a subset of
\eqref{eq:B}, \eqref{eq:bddgrad} holds with the same value of $c_1$ within the level set.
\end{proof}

\begin{lemma}	\label{lem:Lipschitz}
Given any initial point $W^0$, and any $c_2 \geq 0$, there exists
$L_{c_2} > 0$ such that $\|\nabla^2 f(W)\| \leq L_{c_2}$ in the set
$\{W + \bp \mid f(W) \leq f(W^0), \; \|\bp\| \leq c_2\}$.
\end{lemma}
\begin{proof}
Clearly, from the argument in the proof for Lemma \ref{lem:bddgrad},
$\{W + \bp \mid f(W) \leq f(W^0), \; \|\bp\|\leq c_2\}$ is a subset of
the compact set
\[
B\left(0,\sqrt{\frac{2}{\epsilon} f(W^0)} + c_2\right).
\]
Therefore, as a continuous function with respect to $W$, $\|\nabla^2
f(W)\|$ achieves its maximum $L_{c_2}$ in this set.
\end{proof}

Now we provide a convergence guarantee for the SpaRSA algorithm.
\begin{theorem}
	\label{thm:sparsa}
	All accumulation points generated by SpaRSA
	are stationary points.
\end{theorem}
\begin{proof}
We will show that the conditions of \cite[Theorem 1]{WriNF09a} are
satisfied, and thus the result follows.  This theorem states that if
the acceptance condition is
\begin{equation}\label{eq:descent}
		f(W^{t+1}) \leq \max_{i=\max (t-M,0),\dotsc, t} f(W^i)
                - \frac{\sigma\alpha_t}{2}\|W^{t+1} - W^t\|_F^2
\end{equation}
for some
nonnegative integer $M$ and some $\sigma \in (0,1)$, $f$ is Lipschitz
continuously differentiable, the regularizer $c$ defined in
\eqref{eq:def.c} is convex and finite-valued, and $L_I(W)$ of
\eqref{eq:reg-loss} is lower-bounded, then all accumulation points are
stationary.  Clearly, \eqref{eq:decrease} implies the acceptance
condition \eqref{eq:descent}, with $M=0$, and the conditions on
$c(W_1,I)$ and $L_I(W)$ are easily verified.  It remains only to check
Lipschitz continuity of $\nabla f$.  Because the condition
\eqref{eq:decrease} ensures that it is a descent method, all iterates
lie in the set $\{W \mid f(W) \leq f(W^0)\}$.  Thus, by Lemma
\ref{lem:bddgrad}, $f$ has Lipschitz
continuous gradient within this range.  Hence all conditions of
Theorem 1 in \cite{WriNF09a} are satisfied, and the result follows.
\end{proof}

\subsection{Overview of L-BFGS}
\label{subsec:lbfgs}
Before describing our modified L-BFGS algorithm for solving the smooth
problem \eqref{eq:f} obtained after bus selection, we introduce the
original L-BFGS method, following the description from
\cite[Section~7.2]{NocW06}.  Consider the problem
\[
	\min_{W\in \R^d}\quad f(W),
\]
 where $f$ is twice-continuously differentiable.  At iterate $W^t$, L-BFGS
 constructs a symmetric positive definite matrix $B_t$ to approximate
 $\nabla^2 f(W^t)^{-1}$, and the search direction $\bd_t$ is
 obtained as
\begin{equation}
\bd_t = -B_t \nabla f(W^t).
\label{eq:step}
\end{equation}
Given an initial estimate $B_t^0$ at iteration $t$ and a specified
integer $m\geq 0$, we define $m(t) = \min(m, t)$ and construct the
matrix $B_t$ as follows for $t=1,2,\dotsc$:
\begin{align}
	B_t &:= V^T_{t-1}\cdots V^T_{t-m(t)} B_t^0 V_{t-m(t)}\cdots
V_{t-1}
+ \rho_{t-1} s_{t-1} s_{t-1}^T +\nonumber\\
&\qquad \sum_{j=t-m(t)}^{t-2} \rho_{j} V_{t-1}^T\cdots
V_{j+1}^T s_{j} s_{j}^T V_{j+1} \cdots V_{t-1},
\label{eq:Bt}
\end{align}
where for $j\geq 0$, we define
\begin{gather}
	V_j := I - \rho_j \by_j\bs_j^T,\quad
\rho_j := \frac{1}{ \by_j^T\bs_j}, \quad\nonumber\\
\bs_j := W^{j+1}- W^j,\quad
\by_j := \nabla f(W^{j+1}) - \nabla f(W^j).
\label{eq:v}
\end{gather}
The initial matrix
$B_t^0$, for $t \ge 1$, is commonly chosen to be
\begin{equation*}
B_t^0 = \frac{\by_{t-1}^T \bs_{t-1}}{\by_{t-1}^T \by_{t-1}} I.
\end{equation*}
At the first iteration $t = 0$, one usually takes $B_0 = I$, so that
the first search direction $\bd_0$ is the steepest descent direction
$-\nabla f(W^0)$.  After obtaining the update direction $\bd_t$,
L-BFGS conducts a line search procedure to obtain a step size $\eta_t$
satisfying certain conditions, among them the ``sufficient decrease''
or ``Armijo'' condition
\begin{equation} \label{eq:armijo}
f(W^t + \eta_t \bd_t) \leq f(W^t) + \eta_t \gamma \nabla f(W^t)^T
\bd_t,
\end{equation}
where $\gamma \in (0,1)$ is a specified parameter. We assume that the
steplength $\eta_t$ satisfying \eqref{eq:armijo} is chosen via a
backtracking procedure. That is, we choose a parameter $\beta \in
(0,1)$, and set $\eta_t$ to the largest value of $\beta^i$,
$i=0,1,\dotsc$, such that \eqref{eq:armijo} holds.


Note that we use vector notation for such quantities as $\bd_t$,
$\by_j$, $\bs_j$, although these quantities are actually matrices in
our case. Thus, to compute inner products such as $\by_{t}^T \bs_{t}$,
we first need to reshape these matrices as vectors.

\subsection{A Modified L-BFGS Algorithm}
\label{subsec:new-lbfgs}

The key to modifying L-BFGS in a way that guarantees convergence to a
stationary point at a provable rate lies in designing the
modifications so that inequalities of the following form hold, for
some positive scalar values of $a$, $b$, and $\bar{b}$, and for all
vector $\bs_t$ and $\by_t$ defined by \eqref{eq:v} that are used in
the update of the inverse Hessian approximation $B_t$:
\begin{equation} \label{eq:upperlower}
a \|\bs_t\|^2	\leq \by_t^T \bs_t\leq b \|\bs_t\|^2,
\end{equation}
\begin{equation} \label{eq:secondupper}
	\frac{\by_t^T\by_t}{\by_t^T\bs_t}
\leq \bar{b}.
\end{equation}
The average value of the Hessian over the step from $W^t$ to $W^t
+\bs_t$ plays a role in the analysis; this is defined by
\begin{equation} \label{eq:def.Hbar}
	\bar{H}_t := \int^1_0 \nabla^2 f(W^t + t \bs_t) d t.
\end{equation}
When $f$ is strongly convex and twice continuously differentiable, no
modifications are needed: L-BFGS with backtracking line search can be
shown to converge to the unique minimal value of $f$ at a global Q-linear
rate. In this case, the properties \eqref{eq:upperlower} and
\eqref{eq:secondupper} hold when we set $a$ to be the global (strictly
positive) lower bound on the eigenvalues of $\nabla^2 f(W)$ and $b$
and $\bar{b}$ to be the global upper bound on these
eigenvalues. Analysis in \cite{LiuN89a} shows that the eigenvalues of
$B_t$ are bounded inside a strictly positive interval, for all $t$.

In the case of $f$ twice continuously differentiable, but possibly
nonconvex, we modify L-BFGS by skipping certain updates, so as to
ensure that the conditions \eqref{eq:upperlower} and
\eqref{eq:secondupper} are satisfied. Details are given in the
remainder of this section.

We note that conditions \eqref{eq:upperlower} and
\eqref{eq:secondupper} are essential for convergence of L-BFGS not
just theoretically but also empirically. Poor convergence behavior was
observed when we applied the original L-BFGS procedure directly to the
nonconvex 4-layer neural network problem in
Section~\ref{subsec:layers}.

Similar issues regarding poor performance on nonconvex problems are
observed when the full BFGS algorithm is used to solve nonconvex
problems.  (The difference between L-BFGS and BFGS is that for BFGS, in
\eqref{eq:Bt}, $m$ is always set to $t$
and $B_t^0$ is a fixed matrix independent of $t$.)  To ensure
convergence of BFGS for nonconvex problems, \cite{LiF01a} proposed to
update the inverse Hessian approximation only when we are certain
that its smallest eigenvalue after the update is lower-bounded by a
specified positive value.  In particular, those pairs $(\by_j,\bs_j)$
for which the following condition holds: $\tilde\epsilon \|\bs_j\|^2 >
\by_j^T \bs_j$ (for some fixed $\tilde\epsilon>0$) are not used in the
update formula \eqref{eq:Bt}.) Here, we adapt this idea to L-BFGS, by
replacing the indices $t-m(t),\dotsc,t-1$ used in the update formula
\eqref{eq:Bt} by a different set of indices $i^t_1, \dotsc,
i^t_{\hat{m}(t)}$ such that $0 \leq i^t_1 \le \dotsc \le
i^t_{\hat{m}(t)} \leq t-1$, which are the latest $\hat{m}(t)$
iteration indices (up to and including iteration $t-1$) for which the
condition
\begin{equation}
	\bs_j^T \by_j \geq \tilde\epsilon \bs_j^T \bs_j, 
\label{eq:cautious}
\end{equation}
is satisfied. (We define $\hat{m}(t)$ to be the minimum between
$m$ and the number of pairs that satisfy \eqref{eq:cautious}.)
Having determined these indices, we define $B_t$ by
\begin{align}
	B_t & := V^T_{i^t_{\hat{m}(t)}}\cdots V^T_{i^t_1} B_t^0 V_{i^t_1}\cdots
	V_{i^t_{\hat{m}(t)}} + \rho_{i_{\hat{m}(t)}^t}
	\bs_{i_{\hat{m}(t)}^t}\bs_{i_{\hat{m}(t)}^t}^T \nonumber\\
	& + \sum_{j=1}^{\hat{m}(t)-1} \rho_{i^t_j} V_{i^t_{\hat{m}(t)}}^T\cdots
V_{i^t_{j+1}}^T \bs_{i^t_j} \bs_{i^t_j}^T V_{i^t_{j+1}} \cdots
V_{i^t_{\hat{m}(t)}},
\label{eq:Bt_new}
\end{align}
and
\begin{equation}
	B_t^0 = \frac{\by_{i^t_{\hat{m}(t)}}^T
	\bs_{i^t_{\hat{m}(t)}}}{\by_{i^t_{\hat{m}(t)}}^T\by_{i^t_{\hat{m}(t)}}} I.
\label{eq:Bt0_new}
\end{equation}
(When $\hat{m}(t) = 0$, we take $B_t = I$.)  
We show below that, using this rule and the backtracking
line search, we have
\begin{equation}
	\min_{i=0,1,\dotsc,t} \|\nabla f(W^i)\| = O(t^{-1/2}).
	\label{eq:convrate}
\end{equation}
With this guarantee, together with compactness of the level set (see the proof of Lemma \ref{lem:bddgrad}) and that all the algorithm is a descent method so that all iterates stay in this level set, we can prove the
following result.
\begin{theorem} \label{th:conv}
Either we have $\nabla f(W^t)=0$ for some $t$, or else there exists an
accumulation point $\hat{W}$ of the sequence $\{ W^t \}$ that is
stationary, that is, $\nabla f(\hat{W})=0$.
\end{theorem}
\begin{proof}
Suppose that $\nabla f(W^t) \neq 0$ for all $t$. We define a
subsequence ${\cal S}$ of $\{ W^t \}$ as follows:
\[
{\cal S} := \{ \hat{t} \, : \, \| \nabla f(W^{\hat{t}}) \| < \| \nabla f(W^s) \|, \;\; \forall s =0,1,\dotsc, \hat{t}-1 \}.
\]
This subsequence is infinite, since otherwise we would have a strictly
positive lower bound on $\| \nabla f(W^t) \|$, which contradicts
\eqref{eq:convrate}. Moreover, \eqref{eq:convrate} implies that
$\lim_{t \in {\cal S}} \, \| \nabla f(W^t) \| = 0$. Since $\{ W^t
\}_{t \in {\cal S}}$ all lie in the compact level set, this
subsequence has an accumulation point $\hat{W}$, and clearly $\nabla
f(\hat{W})=0$, proving the claim.
\end{proof}

\subsection{Proof of the Gradient Bound}
\label{subsec:convergence}

We now prove the result \eqref{eq:convrate} for the modified L-BFGS
method applied to \eqref{eq:nn}. The proof depends crucially on
showing that the bounds \eqref{eq:upperlower} and
\eqref{eq:secondupper} hold for all vector pairs $(\bs_j,\by_j)$ that
are used to define $B_t$ in \eqref{eq:Bt_new}.
\begin{theorem}
\label{thm:lbfgs}
Given any initial point $W^0$, using the modified L-BFGS algorithm
discussed in Section~\ref{subsec:new-lbfgs} to optimize \eqref{eq:nn},
then there exists $\delta > 0$ such
\begin{equation}
1 \geq \frac{-\nabla f(W^t)^T \bd_t}{\|\nabla f(W^t)\|\|\bd_t\|} \geq
\delta, \quad t=0,1,2,\dotsc.
\label{eq:angle}
\end{equation}
Moreover, there exist $M_1,M_2$ with $M_1 \geq M_2 > 0$ such that 
\begin{equation}
M_2 \|\nabla f(W^t)\| \leq \|\bd_t\| \leq M_1 \|\nabla f(W^t)\|,
\quad  t=0,1,2,\dotsc.
\label{eq:bound}
\end{equation}
\end{theorem}
\begin{proof}
We first show that for all $t > 0$, the following descent condition
holds:
\begin{equation}
f(W^t) \leq f(W^{t-1}),
\label{eq:descent2}
\end{equation}
implying that 
\begin{equation}
W^t \in \{W \mid f(W) \leq f(W^0)\}.
\label{eq:levelset}
\end{equation}
To prove \eqref{eq:descent2}, for the case that $t = 0$ or $\hat{m}(t) =
0$, it is clear that $\bd_t = -\nabla f(W^t)$ and thus the condition
\eqref{eq:armijo} guarantees that \eqref{eq:descent2} holds.  We now
consider the case $\hat{m}(t) > 0$.  From \eqref{eq:Bt_new}, since
\eqref{eq:cautious} guarantees $\rho_{i^t_j} \geq 0$ for all $j$, we
have that $B_t$ is positive semidefinite.  Therefore, \eqref{eq:step}
gives $\nabla f(W^t)^T \bd_t \leq 0$, which together with
\eqref{eq:armijo} implies \eqref{eq:descent2}.

Next, we will show \eqref{eq:upperlower} and \eqref{eq:secondupper}
hold for all pairs $(\bs_j,\by_j)$ with
$j=i^t_1,\dotsc,i^t_{\hat{m}(t)}$. 
The
left inequality in \eqref{eq:upperlower} follows directly from
\eqref{eq:cautious}, with $a = \tilde{\epsilon}$. We now prove the
right inequality of \eqref{eq:upperlower}, along with
\eqref{eq:secondupper}. Because \eqref{eq:levelset} holds, we have
from Lemma~\ref{lem:Lipschitz} that $\bar{H}_t$ defined by
\eqref{eq:def.Hbar} satisfies
\begin{equation}
\|\bar{H}_t\| \leq L_{c_2}, \quad t=0,1,2,\dotsc.
\label{eq:bddH}
\end{equation}
From $\by_t = \bar{H}_t \bs_t$, \eqref{eq:bddH}, and
\eqref{eq:cautious}, we have for all $t$ such that
\eqref{eq:cautious} holds that
\begin{equation}
\frac{\|\by_t\|^2}{\by_t^T \bs_t} \leq \frac{L_{c_2}^2
  \|\bs_t\|^2}{\by_t^T \bs_t} \leq \frac{L_{c_2}^2}{\tilde\epsilon},
\label{eq:upper2}
\end{equation}
which is exactly \eqref{eq:secondupper} with $\bar{b} =
L_{c_2}^2/\tilde{\epsilon}$.

From $\by_t = \bar{H}_t \bs_t$, the Cauchy-Schwarz inequality, and
\eqref{eq:bddH}, we get
\begin{equation*}
\by_t^T \bs_t \leq \|\by_t\|\|\bs_t\| \leq
\|\bs_t\|\|\bar{H}_t\|\|\bs_t\| =L_{c_2}\|\bs_t\|^2,
\end{equation*}
proving the right inequality of \eqref{eq:upperlower}, with $b =
L_{c_2}$.


Now that we have shown that \eqref{eq:upperlower} and
\eqref{eq:secondupper} hold for all indices
$i^t_1,\dotsc,i^t_{\hat{m}(t)}$, we can follow the proof in
\cite{LiuN89a} to show that there exist $M_1 \geq M_2 > 0$ such that
\begin{equation}
M_1 I \succeq B_t \succeq M_2 I, \quad \mbox{for all $t$.}
\label{eq:Bbound}
\end{equation}
The rest of the proof is devoted to showing that this bound
holds. Having proved this bound, the results \eqref{eq:bound} and
\eqref{eq:angle} (with $\delta = M_2 / M_1$) follow directly from the
definition \eqref{eq:step} of $\bd_t$.

To prove \eqref{eq:Bbound}, we first bound $B_t^0$ defined in
\eqref{eq:Bt0_new}.  This bound will follow if we can prove a bound on
$\by_t^T\bs_t / \|\by_t\|^2$ for all $t$ satisfying
\eqref{eq:cautious}.  Clearly when $\hat{m}(t) = 0$, we have $B_t^0 =
I$, so there are trivial lower and upper bounds.  For $\hat{m}(t) >
0$, \eqref{eq:secondupper} implies a lower bound of $1/\bar{b} =
\tilde{\epsilon}/L_{c_2}^2$. For an upper bound, we have from 
\eqref{eq:cautious}  that
\[
\tilde{\epsilon} \| \bs_j\|^2 \le  \| \bs_j \| \| \by_j \| \;\;
\Rightarrow \;\; \| \bs_j \| \le \frac{1}{\tilde{\epsilon}} \| \by_j \|.
\]
Hence, from \eqref{eq:Bt0_new}, we have
\[
\left\| B_t^0 \right\| = \frac{\left|\by_{i^t_{\hat{m}(t)}}^T
  \bs_{i^t_{\hat{m}(t)}}\right|}{\left\| \by_{i^t_{\hat{m}(t)}}\right\|^2 } \le
\frac{\left\| \bs_{i^t_{\hat{m}(t)}} \right\|}{\left\| \by_{i^t_{\hat{m}(t)}}\right\|} \le
\frac{1}{\tilde{\epsilon}}.
\]


Now we will prove the results by working on the inverse of $B_t$.
Following \cite{LiuN89a}, the inverse of $B_t$ can be obtained by
\begin{align}
H_t^{(0)} &= (B_t^0)^{-1},\nonumber\\ H_t^{(k+1)} &= H_t^{(k)} -
\frac{H_t^{(k)}\bs_{i^t_k} \bs_{i^t_k}^T H_t^{(k)}}{\bs_{i^t_k}^T
  H_t^{(k)}\bs_{i^t_k}} \nonumber\\ &\qquad
+\frac{\by_{i^t_k}\by_{i^t_k}^T }{\by_{i^t_k}^T \bs_{i^t_k}}, \;\; 
k=0,\dots, \hat{m}(t)-1,
\label{eq:inverse}\\
B_t^{-1} &= H_t^{\hat{m}(t)}.
\nonumber
\end{align}
Therefore, we can bound the trace of $B_t^{-1}$ by using \eqref{eq:upper2}.
\begin{align}
\trace(B_t^{-1}) &\leq \trace((B_t^0)^{-1}) + \sum_{k=0}^{\hat{m}(t)-1}
\frac{\by_{i^t_k}^T\by_{i^t_k}}{\by_{i^t_k}^T \bs_{i^t_k}}
\nonumber\\ &\leq \trace((B_t^0)^{-1}) + \hat{m}(t) \frac{L_{c_2}^2}{\tilde{\epsilon}}.
\label{eq:uppB}
\end{align}
This together with the fact that $B_t^{-1}$ is positive-semidefinite
and that $B_t^0$ is bounded imply that there exists $M_2>0$ such that 
\[
\| B_t^{-1} \| \le \trace(B_t^{-1}) \le M_2^{-1},
\]
which implies that $B_t \succeq M_2 I$, proving the right-hand
inequality in \eqref{eq:Bbound}. (Note that this upper bound for the
largest eigenvalue also applies to $H_t^{(k)}$ for all
$k=0,1,\dotsc,\hat{m}(t)-1$.)

For the left-hand side of \eqref{eq:Bbound}, we have from the
formulation for \eqref{eq:inverse} in \cite{LiuN89a} (see \cite{Pea69a} for a derivation) and the upper bound $\|H_t^{(k)}\| \le
M_2$ that
\begin{align*}
\det(B_t^{-1})
		&=\det((B_t^0)^{-1}) \prod_{j=0}^{\hat{m}(t)-1}
		\frac{\by_{i^t_k}^T\bs_{i^t_k}}{\bs_{i^t_k}^T \bs_{i^t_k}}
		\frac{\bs_{i^t_k}^T\bs_{i^t_k}}{\bs_{i^t_k}^T H_t^{(k)}
		\bs_{i^t_k}}\\
		&\geq \det((B_t^0)^{-1}) \left(\frac{\tilde\epsilon}{M_2} \right)^{\hat{m}(t)} \\
& \geq \bar{M}_1^{-1},
\end{align*}
for some $\bar{M}_1>0$.  
Since the eigenvalues of $B_t^{-1}$ are upper-bounded by $M_2^{-1}$,
it follows from the positive lower bound on $\det (B_t^{-1})$ that
these eigenvalues are also lower-bounded by a positive number. The
left-hand side of \eqref{eq:Bbound} follows.
\end{proof}

\begin{corollary}
Given any initial point $W^0$, if we use the algorithm discussed in
Section~\ref{subsec:new-lbfgs} to solve \eqref{eq:nn}, then the bound
\eqref{eq:convrate} holds for the norms of the gradients at the
iterates $W^0,W^1, \dotsc$.
\end{corollary}
\begin{proof}
First, we lower-bound the step
size obtained from the backtracking line search procedure.  Consider
any iterate $W^t$ and the generated update direction $\bd_t$ by the
algorithm discussed in Section \ref{subsec:new-lbfgs}.  From
Theorem~\ref{thm:lbfgs}, we have
\[
\| \bd_t \|^2 \le M_1 \| \bd_t \| \| \nabla f(W^t ) \| \le -\frac{M_1}{\delta}
\nabla f(W^t)^T \bd_t.
\]
Thus by using Taylor's theorem, and the uniform upper bound on $\|
\nabla^2 f(W)\|$ in the level set defined in
Lemma~\ref{lem:Lipschitz}, we have for any value of $\eta$
\begin{align*}
	f(W^t + \eta \bd_t) &\le f(W^t) + \eta \nabla f(W^t)^T \bd_t +
		\frac{L_{c_2}\eta^2}{2}\|\bd_t\|^2\\
		&\leq f(W^t) + \eta \nabla f(W^t)^T \bd_t \left(1 -
		\frac{L_{c_2}M_1 \eta}{2\delta} \right).
\end{align*}

Therefore, since $\nabla f(W^t)^T \bd_t < 0$, \eqref{eq:armijo} holds
whenever
\begin{equation*}
		1 - \eta \frac{L_{c_2}M_1}{2\delta} \geq \gamma \;\;
\Leftrightarrow \;\; \eta \le \bar{\eta} := \frac{2(1-\gamma)\delta}{L_{c_2}M_1}.
\end{equation*}
Because the backtracking mechanism decreases the candidate stepsize by
a factor of $\beta \in (0,1)$ at each attempt, it will ``undershoot''
$\bar{\eta}$ by at most a factor of $\beta$, so we have
\begin{equation}
\label{eq:stepsize}
\eta_t \geq \min (1, \beta \bar{\eta}), \quad \mbox{for all $t$.}
\end{equation}
From \eqref{eq:armijo},
Theorem~\ref{thm:lbfgs}, and \eqref{eq:stepsize} we have that
\begin{align}
f(W^{t+1}) &\leq f(W^t) + \eta_t \gamma \nabla f(W^t)^T \bd_t
		\nonumber\\
		&\leq f(W^t) - \eta_t \gamma \delta \|\nabla
		f(W^t)\| \| \bd_t \| \nonumber\\
		&\leq f(W^t) - \eta_t \gamma \delta M_2 \|\nabla
		f(W^t)\|^2\nonumber\\
		&\leq f(W^t) - \hat{\eta} \|\nabla f(W^t)\|^2.
		\label{eq:tosum}
\end{align}
where $\hat{\eta} := \min (1, \beta \bar{\eta})\gamma \delta M_2$.
Summing \eqref{eq:tosum} over $t=0,1,\dotsc,k$, we get
\begin{align*}
		\min_{0 \leq t \leq k} \|\nabla f(W^t)\|^2 &\leq
                \frac{1}{k+1} \sum_{t=0}^k \|\nabla f(W^t)\|^2\\ &\leq
                \frac{1}{k+1}\frac{1}{\hat{\eta}} \sum_{t=0}^k (f(W^t) -
                f(W^{t+1}))\\ 
&\leq
                \frac{1}{k+1}\frac{1}{\hat{\eta}} (f(W^0) -
                f(W^{k+1}))\\ 
&\leq \frac{1}{k+1} \frac{1}{\hat{\eta}}
                (f(W^0) - f^*) = O(1/k).
\end{align*}
where $f^*$ is the optimal function value in \eqref{eq:f}, which is
lower-bounded by zero.  By taking square roots of both sides, the
claim \eqref{eq:convrate} follows.
\end{proof}
\subsection{Neural Network Initialization}
Initialization of the neural network training is not trivial. The
obvious initial point of $W_j=0$, $j \in [N+1]$ has $\nabla_{W_j}
f=0$, $j \in [N+1]$ (as can be seen via calculations with
\eqref{eq:f}, \eqref{eq:tanh}, and \eqref{eq:grads}), so is likely a
saddle point. A gradient-based step will not move away from such a
point.
Rather, we start from a random point close to the origin. Following a
suggestion from a well-known online
tutorial,\footnote{\url{http://deeplearning.net/tutorial/mlp.html\#weight-initialization}}
we choose all elements of each $W_j$ uniformly, randomly, and
identically distributed from the interval
$[-a\sqrt{6}/\sqrt{d_{j-1}+d_{j}},a\sqrt{6}/\sqrt{d_{j-1}+d_{j}}]$, where
$a=1$.  We experimented with smaller values of $a$, when setting $a=1$
leads to slow convergence in the training error (which is an indicator
that the initial point is not good enough), by starting with $a =
10^{-t}$ for some non-negative integer $t$.  We keep trying smaller
$t$ until either the convergence is fast enough, or the resulting
solution has high training errors and the optimization procedure
terminates early. In the latter case, we then set $t\leftarrow t+1$,
choose new random points from the interval above for the new value of
$a$, and repeat. 


\subsection{Additional Experiment on Using More Layers in the Neural
Networks}
\label{subsec:layers}

We now examine the effects of adding more hidden layers to the neural
network. As a test case, we choose the 57-bus case, with ten
pre-selected PMU locations, at nodes $[1, 2, 17, 19, 26, 39, 40, 45,
	46, 57]$. (These were the PMUs selected by the greedy heuristic in
\cite[Table~III]{KimW16a}.) We consider three neural network
configurations. The first is the single hidden layer of 200 nodes
considered above. The second contains two hidden layers, where the
layer closer to the input has 200 nodes and the layer closer to the
output has 100 nodes. The third configuration contains four hidden
layers of 50 nodes each. For this last configuration, when we solved
the training problem with L-BFGS, the algorithm frequently required
modification to avoid negative-curvature directions. (In that sense,
it showed greater evidence of nonconvexity.)

Figure~\ref{fig:multiple_new} shows the training error and test error
rates as a function of training time, for these three
configurations. The total number of variables in each model is shown
along with the final training and test error rates in
Table~\ref{tbl:multiple}. The training error ultimately achieved is
smaller for the multiple-hidden-layer configurations than for the
single hidden layer.  
However, the single hidden layer still has a slightly better test
error. This suggests that the multiple-hidden-layer models may have
overfit the training data. (A further indication of overfitting is
that the test error increases slightly for the four-hidden-layer
configuration toward the end of the training interval.) This test is
not definitive, however; with a larger set of training data, we may
find that the multiple-hidden-layer models give better test errors.
\begin{table}
\centering
\caption{Performance of different number of layers in the neural network model.}
\label{tbl:multiple}
\begin{tabular}{l|rrr}
\#layers & \#variables & Test error & Training error\\
\hline
1 & 19,675 & 4.15\% & 2.36\% \\
2 & 32,275 & 5.87\% & 1.50\% \\
4 & 12,625 & 6.83\% & 2.02\%
\end{tabular}
\end{table}

\begin{figure}
	\centering
\subfloat[Training error]{
		\includegraphics[width=.48\linewidth,height=1in]{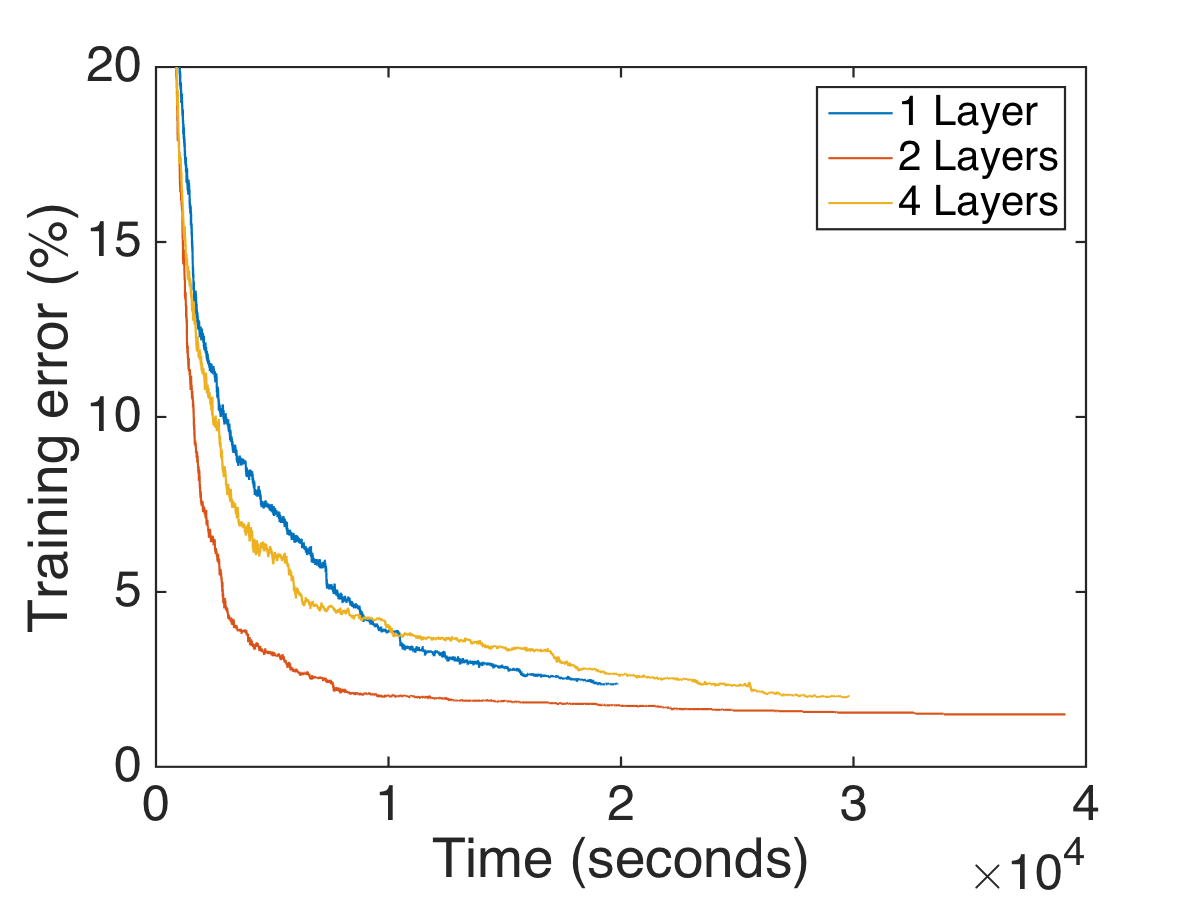}}
	\subfloat[Test error]{
		\includegraphics[width=.48\linewidth,height=1in]{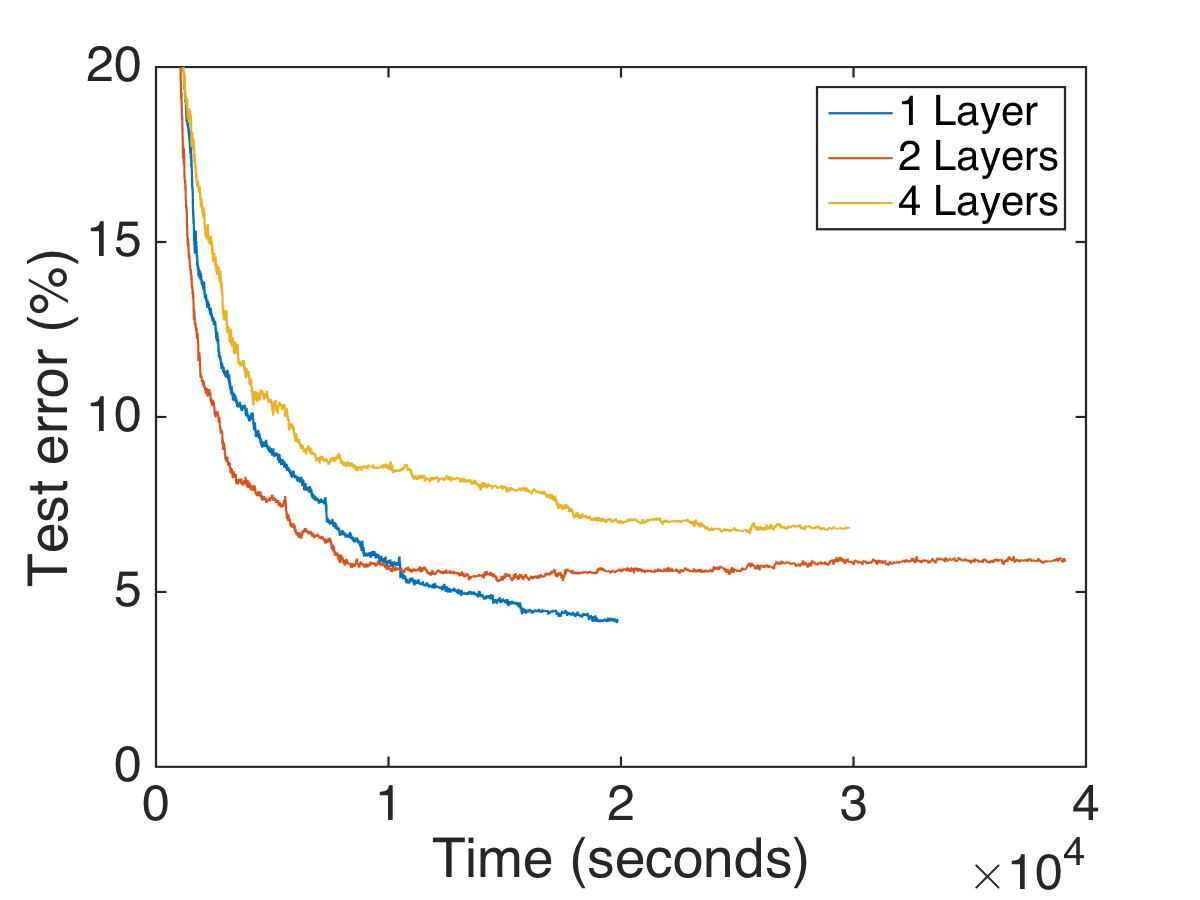}}
	\caption{Comparison between 1, 2 and 4 layers. We show training and test error v.s. running time.}
	\label{fig:multiple_new}
\end{figure}
}
\end{document}